\documentclass[a4paper,11pt,twoside]{article}
\usepackage{CJK}
\usepackage{latexsym,bm}
\usepackage[tbtags]{amsmath}
\usepackage{amssymb,amsmath}

\usepackage{fancyhdr}
\usepackage{cite}
\usepackage{tabularx}
\usepackage{booktabs}
\usepackage{dcolumn}

\usepackage{multirow}

\usepackage{graphicx}

\usepackage{wrapfig}
\usepackage{multicol}
\usepackage{verbatim}
\usepackage{enumerate,color}
\def\firstpage{1}                           

\def\shorttitle{Isotropic polynomial invariants of the Hall tensor}           
\def\shortauthor{Jinjie LIU,\quad  Weiyang DING,\quad  Liqun QI,\quad  Wennan ZOU}
\usepackage{Shangda0616}                    

\newcommand{\supercite}[1]{\!\!\textsuperscript{\cite{#1}}} 
\newtheorem{lemma}{\indent Lemma}

\begin{document} 


\title{{\large  \textbf{Isotropic polynomial invariants of the Hall tensor}}
\thanks{\footnotesize{Received (date) / Revised (date)
\newline
Project supported by HKBU RC's Start-up Grant for New Academics, the Hong Kong Research Grant Council (Grant Nos. PolyU 15302114, 15300715, 15301716 and 15300717), and the National Natural Science Foundation of China (Grant No. 11372124)
\newline
Corresponding author, Liqun QI, Professor,
E-mail:\;maqilq@polyu.edu.hk}}}
\author{\small{Jinjie LIU$^1$,\quad Weiyang DING$^{2}$,\quad  Liqun QI$^1$,\quad Wennan ZOU$^3$} 
\\[2mm]
\footnotesize {1. Department of Applied Mathematics, The Hong Kong Polytechnic University, Hong Kong, S. A. R. China;} 
\\
\footnotesize {2. Department of Mathematics, Hong Kong Baptist University, Hong Kong, S. A. R. China;} 
\\
\footnotesize {3. Institute for Advanced Study, Nanchang University, Nanchang 330031, P. R. China}
 } 

\maketitle 
\thispagestyle{first}
\footnotesize
\begin{abstract}
\noindent \textbf{Abstract~~~}
  The Hall tensor emerges from the study of the Hall effect, an important magnetic effect observed in electric conductors and semiconductors. The Hall tensor is third order and three dimensional, whose first two indices are skew-symmetric. In this paper, we investigate the isotropic polynomial invariants of the Hall tensor by connecting it with a second order tensor via the third order Levi-Civita tensor.
  We propose a minimal isotropic integrity basis with 10 invariants for the Hall tensor. Furthermore, we prove that this minimal integrity basis is also an irreducible isotropic function basis of the Hall tensor.
\\[2mm]
\textbf{Key words~~~}isotropic polynomial invariants, irreducibility, function basis, integrity basis, Hall tensor
\\[2mm]
\textbf{Chinese Library Classification~~~}O29
\\
\textbf{2010 Mathematics Subject Classification~~~}15A69, 15A72, 15A90
\end{abstract}

\section*{\small Nomenclature}
\footnotesize

\setlength{\columnsep}{5mm}
\begin{multicols}{2}
\begin{zk}{$a$\quad}
\item[$\bf I$,] second order identify tensor;
\item[$\bm \varepsilon$,] permutation tensor(i.e., Levi-Civita tensor) with components $\varepsilon_{ijk}$ in three dimensions;
\item[$\bf Q$,] orthogonal tensor with components $q_{ij}$;
\item[$\langle {\bf Q}\rangle {\bf A}$,]\ a second order tensor ${\bf A}$ under an orthogonal transformation;
\item[$\mathcal K$,] a Hall tensor with components $k_{ijk}$;
\item[${\bf e}_i \otimes {\bf e}_j \otimes {\bf e}_k $,]\ orthonormal base in three dimensions;
\item[${\rm tr}\, {\bf A}$, ${\bf A}^{\top}$,]\ trace and transpose, respectively, of a second order tensor $\bf A$;
\item[${\rm det}\, {\bf A}$,]\ determinant of a second order tensor $\bf A$;
\item[${\bm \varepsilon}{\mathcal K}$,]\  second order tensor with components $\varepsilon_{kli} k_{klj}$ in three dimensions;
\item[${\bm \varepsilon}{\bf A}$,]\  third order tensor with components $\varepsilon_{ijl} a_{lk}$ in three dimensions;
\item[$\mathbb{R}^n$,]\  the real number field with dimension $n$;
\item[$\mathbb{R}^{m\times n}$,]\  $m-$by$-n$ matrix on the real number field.
\end{zk}
\end{multicols}

\section{Introduction}
\small
Tensor function representation theory is an essential topic in continuum mechanics, which focuses on the tensor invariants under coordinate transformations. Since tensor invariants often reveal more intrinsic information of materials than tensor components, the complete and irreducible representation for invariant tensor functions plays a key role in modeling nonlinear constitutive equations in both theoretical and applied physics. Such representation prescribes the number and the type of scalar variables required in the constitutive equations. These representations are efficient in the process of describing the physical behavior of anisotropic materials, because the invariant conditions are dominant and no other simple methods are able to determine such information. There are plenty of research works on this topic from the last century\supercite{Boehler1977, BKO94, Pennisi87, Smith71, SmithBao1997, Spencer1970, Wang70, Zheng93, Zheng94, Zheng96, ZhengBetten95}. For instance, the minimal integrity basis and irreducible function basis of a second order tensor in three dimensions were well studied by Wang\supercite{Wang70}, Smith\supercite{Smith71}, Boehler\supercite{Boehler1977}, Pennisi and Trovato\supercite{Pennisi87}, and Zheng\supercite{Zheng93}. Furthermore, Boehler et al.\supercite{BKO94} investigate polynomial invariants for the elasticity tensor in 1993. Some very recent works are devoted to minimal integrity bases and irreducible function bases for third order and fourth order tensors\supercite{ChenQiZou17, CLQZZ18}.

The tensor function representation theory is also closely related to the classical invariant theory in algebraic geometry\supercite{AuffrayRopars16, Hilbert93, OliveAuffray14, Sturmfels}. One of the most famous approaches for computing the complete invariant basis was first introduced by Hilbert\supercite{Hilbert93}. In 2017, Olive, Kolev, and Auffray\supercite{OliveKolevAuffray17} studied the minimal integrity basis with $297$ invariants of the fourth order elasticity tensors successfully via the approaches from the algebraic geometry.

The Hall effect is an important magnetic effect observed in electric conductors and semiconductors \supercite{Haussuhl08}.  It was discovered in 1879 by and named after Edwin Hall \supercite{Hall1879}.
When an electric current density ${\bf J}$ is flowing through a plate and the plate is simultaneously immersed in a magnetic field ${\bf H}$ with a component transverse to the current, the electric field strength ${\bf E}$ is proportional to current density and magnetic field strength
\begin{align*}
E_i = k_{ijk} J_j H_k,
\end{align*}
where the third order tensor ${\mathcal K}$ with components $k_{ijk}$ is called the Hall tensor \supercite{Haussuhl08}.
We note that the representation of the Hall tensor under any orthonormal basis satisfies $k_{ijk} = -k_{jik}$ for all $i,j,k=1,2,3$, since the Onsager relation for transport processes with time reversal is valid.
The Hall tensor is essential for describing the electromagnetic induction. Therefore, it is significant to investigate the minimal integrity basis and irreducible function basis of the Hall tensor.

This paper is devoted to the invariants of the Hall tensor and organized as follows.
We first briefly review some basic definitions and the relationship between the integrity basis and the function basis of a tensor in Section 2. Then we build a connection between the invariants of a Hall tensor and that of  a second order tensor, which is important for the subsequent contents. Furthermore, we propose a minimal isotropic integrity basis with 10 isotropic invariants of the Hall tensor in Section 3.  In Section 4, we proved that the  minimal integrity basis with 10 invariants of the Hall tensor is also its irreducible function basis.
Finally, we draw some concluding remarks and raise one further question in Section 5.


\section{Preliminaries}

Denote $\mathcal{A}$ as an $m$th order tensor represented by $a_{i_1i_2\cdots i_m}$ under some orthonormal coordinate.
A scalar-valued tensor function $f(\mathcal{A})$ is called an isotropic invariant of ${\mathcal A}$ if it is invariant under any orthogonal transformations, including rotations and reflections, i.e.,
$$f(\langle {\bf Q} \rangle \mathcal{A})=f(\mathcal{A}),$$
or equivalently expressed by
$$f(q_{i_1 j_1} q_{i_2 j_2} \dots q_{i_m j_m}a_{j_1 j_2 \dots j_m} ) = f(a_{i_1 i_2 \dots i_m}),$$
where $\bf Q$ is an orthogonal tensor (${\bf Q}^{\top}{\bf Q}={\bf Q}{\bf Q}^{\top}={\bf I}$) with components $q_{ij}$. If $f(\mathcal{A})$ is only invariant under rotations, i.e., $f(\langle {\bf Q} \rangle \mathcal{A})=f(\mathcal{A})$
for any orthogonal tensor $\bf Q$ with ${\rm det}\, {\bf Q}=1$, then it is called a hemitropic invariant of tensor $\mathcal{A}$.
Furthermore, if $f(\mathcal{A})$ is a polynomial, then it is called a polynomial invariant of ${\mathcal A}$.  In the subsequence, invariants always stand for polynomial invariants unless specific remarks are made there.

For any second order tensor, the hemitropic invariants and the isotropic invariants are equivalent, since it keeps unaltered under the central inversion $-{\bf I}$ \supercite{Zheng94}. Nevertheless, any isotropic polynomial invariant of a third order tensor has to be the summation of several even order degree polynomials.

We then briefly review the definitions and properties of (minimal) integrity bases and (irreducible) function bases of a tensor.

\begin{Dingli}[Definition 1 (integrity basis).]\label{def1}
Let $\Psi = \{\psi_1,\psi_2,\dots,\psi_r\}$ be a set of  isotropic (or hemitropic, respectively) invariants of a tensor $\mathcal{A}$.
\begin{enumerate}[{\rm (1)}]
  \item $\Psi$ is said to be {\bf polynomial irreducible} if none of $\psi_1, \psi_2, \dots, \psi_r$ can be expressed by a polynomial of the remainders;
  \item $\Psi$ is called an  isotropic (or hemitropic, respectively)  {\bf integrity basis} if any isotropic (or hemitropic, respectively) invariant of $\mathcal{A}$ is expressible by a polynomial of $\psi_1, \psi_2, \dots, \psi_r$;
  \item $\Psi$ is called an isotropic (or hemitropic, respectively) {\bf minimal integrity basis} if it is polynomial irreducible and an isotropic (or hemitropic, respectively) integrity basis.
\end{enumerate}
\end{Dingli}
\begin{Dingli}[Definition 2 (function basis).]\label{def2}
Let $\Psi = \{\psi_1,\psi_2,\dots,\psi_r\}$ be a set of  isotropic (or hemitropic, respectively) invariants of a tensor $\mathcal{A}$.
\begin{enumerate}[{\rm (1)}]
  \item An invariant in $\Psi$ is said to be  {\bf functionally irreducible} if it cannot be expressed by a single-valued function of the remainders,
  $\Psi$ is said to be  functionally irreducible if all of $\psi_1, \psi_2, \dots, \psi_r$ are functionally irreducible;
  \item $\Psi$ is called an isotropic (or hemitropic, respectively)  {\bf function basis} if any isotropic (or hemitropic, respectively) invariant of $\mathcal{A}$ is expressible by a function  of $\psi_1, \psi_2, \dots, \psi_r$;
  \item $\Psi$ is called an isotropic (or hemitropic, respectively) {\bf irreducible function basis} if it is functionally irreducible and is an isotropic (or hemitropic, respectively) function basis.
\end{enumerate}
\end{Dingli}

It is straightforward to verify by definitions that an isotropic (or hemitropic, respectively) integrity basis is an isotropic (or hemitropic, respectively) function basis, but the converse is incorrect in general.  Thus, the number of invariants in an isotropic (or hemitropic, respectively) irreducible function basis is no greater than the number of invariants in an isotropic (or hemitropic, respectively) minimal integrity basis. For instance, the number of irreducible function basis of a third order traceless symmetric tensor is 11 which is less than 13, the number of invariants in its minimal integrity basis\supercite{CLQZZ18}.

Particularly, it has been proved that the number of invariants of each degree in an isotropic (or hemitropic, respectively) minimal integrity basis is always fixed \supercite{OliveKolevAuffray17}.   Nevertheless, it is still unclear whether the number of invariants of an irreducible function basis is fixed.

\section{Minimal integrity basis of the Hall tensor}

Let  ${\mathcal K}$ be a Hall tensor represented by $k_{ijk}$ under an orthonormal basis ${\bf e}_i \otimes {\bf e}_j \otimes {\bf e}_k$.   Define a second order tensor ${\bf A}$ accordingly, with components $a_{ij}$ under this orthonormal basis, by the tensor product operation
$${\bf A}:=\textstyle\frac{1}{2}{\bm \varepsilon}{\mathcal{K}},$$
or equivalently
\begin{align*}
(a_{ij}){\bf e}_i \otimes {\bf e}_j = (\textstyle\frac{1}{2} \varepsilon_{kli} k_{klj}){\bf e}_i \otimes {\bf e}_j ,
\end{align*}
where ${\bm \varepsilon}$ is the third order Levi-Civita tensor.   Conversely, the Hall tensor can also be expressed with this second order tensor by
$$\mathcal{K}:={\bm \varepsilon}{\bf A},$$
or equivalently
\begin{align*}\label{eq_hall}
  (k_{ijk}){\bf e}_i \otimes {\bf e}_j \otimes {\bf e}_k = (\varepsilon_{ijl} a_{lk}){\bf e}_i \otimes {\bf e}_j \otimes {\bf e}_k.
\end{align*}
There are nine independent components in a Hall tensor $\mathcal{K}$ due to the anti-symmetry of the first two indices of the components in a Hall tensor. Without loss of generality, denote the nine independent components of the Hall tensor ${\mathcal K}$ as:
$$k_{121}, k_{122}, k_{123}, k_{131}, k_{132}, k_{133}, k_{231}, k_{232}, k_{233}.$$
Then under a right-handed coordinate, the representation of the associated second order tensor can be written into a matrix form mathematically:
\begin{align*}
\begin{pmatrix}
  k_{231} & k_{232} & k_{233} \\ -k_{131} & -k_{132} & -k_{133} \\ k_{121} & k_{122} & k_{123}
\end{pmatrix} \in \mathbb{R}^{3 \times 3}.
\end{align*}

The following theorem reveals the relationship between the invariants of the Hall tensor and the ones of its associated second order tensor.

\begin{lemma}[Theorem 1.]\label{thm_invariant}
Let $\mathcal{K}$ be a Hall tensor with components $k_{ijk}$. We use ${\bf A}(\mathcal{K})$ to denote the associated second order tensor of ${\cal K}$.
\begin{enumerate}[{\rm (1)}]
  \item Any isotropic invariant of $\mathcal{K}$ is an isotropic invariant of ${\bf A}(\mathcal{K})$;
  \item Any isotropic invariant of ${\bf A}(\mathcal{K})$ with even degree is an isotropic invariant of $\mathcal{K}$, and any isotropic invariant of ${\bf A}(\mathcal{K})$ with odd degree is an hemitropic invariant of $\mathcal{K}$
\end{enumerate}
\end{lemma}
\begin{proof}
(1) An isotropic invariant $f(\mathcal{K})$ of the Hall tensor $\mathcal{K}$ is also a polynomial function of its associated second order tensor ${\bf A}(\mathcal{K})$, denoted by $g({\bf A}) := f({\bm \varepsilon}{\bf A})$. Now we need to show that $g({\bf A})$ is an isotropic invariant of  ${\bf A}(\mathcal{K})$. Let ${\bf Q}$ be any orthogonal tensor. According to the definition of  isotropic invariants, we have
$$
  f(\langle{\bf Q}\rangle\mathcal{K}) = f(\mathcal{K}) = f({\bm \varepsilon}{\bf A})=  g({\bf A}).
$$
Make use of the equality $\langle{\bf Q}\rangle{\bm \varepsilon} = ({\rm det}\, {\bf Q}){\bm \varepsilon}$, then
$$
  f(\langle{\bf Q}\rangle\mathcal{K}) =  f(\langle{\bf Q}\rangle({\bm \varepsilon}{\bf A})) =  f(\langle{\bf Q}\rangle{\bm \varepsilon}\langle{\bf Q}\rangle{\bf A})= f(({\rm det}\, {\bf Q}){\bm \varepsilon}\langle{\bf Q}\rangle{\bf A})=g(({\rm det}\, {\bf Q})\langle{\bf Q}\rangle{\bf A}).
$$
Since an isotropic invariant of a third order tensor must be an even function, thus
$$g(({\rm det}\, {\bf Q})\langle{\bf Q}\rangle{\bf A})=g(\langle{\bf Q}\rangle{\bf A}).$$
Hence, $g(\langle{\bf Q}\rangle{\bf A})=  g({\bf A})$, i.e., $g({\bf A})$ is an isotropic invariant of  ${\bf A}$.

(2) Denote an invariant of ${\bf A}(\mathcal{K})$ as $g({\bf A})$. It is also a polynomial of the Hall tensor $\mathcal{K}$, denoted by $f({\cal K}) := g(\frac{1}{2}{\bm \varepsilon}{\cal K})$. For any orthogonal tensor $\bf Q$, since $g({\bf A})$ is an invariant, we know
$$
  g(\langle{\bf Q}\rangle{\bf A}) = g({\bf A}) = g({\bm \varepsilon}\mathcal{K})=  f(\mathcal{K}).
$$
Recall that $\langle{\bf Q}\rangle{\bm \varepsilon} = ({\rm det}\, {\bf Q}){\bm \varepsilon}$. Then
$$
  g(\langle{\bf Q}\rangle{\bf A}) = g(\langle{\bf Q}\rangle({\bm \varepsilon}\mathcal{K})) =  g(\langle{\bf Q}\rangle{\bm \varepsilon}\langle{\bf Q}\rangle\mathcal{K})
  = g(({\rm det}\, {\bf Q}){\bm \varepsilon}\langle{\bf Q}\rangle\mathcal{K} = f(({\rm det}\, {\bf Q})\langle{\bf Q}\rangle\mathcal{K}).
$$
Hence, when $g({\bf A})$ is an invariant of even degree, we have $f(\langle{\bf Q}\rangle\mathcal{K})=f(\mathcal{K})$ for any orthogonal tensor $\bf Q$. That is, $f(\mathcal{K})$ is an isotropic invariant of the Hall tensor $\mathcal{K}$.

When $g({\bf A})$ is an invariant of odd degree, only for orthogonal tensor $\bf Q$ satisfying ${\rm det}\, {\bf Q} =1$, it holds that $f(\langle{\bf Q}\rangle\mathcal{K})=f(\mathcal{K})$, which means that $f(\mathcal{K})$ is an hemitropic invariant of the Hall tensor $\mathcal{K}$. The proof is completed.
\end{proof}

Hence, we can construct an integrity basis for a Hall tensor from the integrity basis of its associated second order tensor. For the associated second order tensor ${\bf A}(\mathcal{K})$, we split it into ${\bf A}(\mathcal{K}) = {\bf T} + {\bf W}$, where ${\bf T}$ is symmetric with components $t_{ij} = \frac{1}{2}(a_{ij} + a_{ji})$ and ${\bf W}$ is skew-symmetric with components $w_{ij} = \frac{1}{2}(a_{ij} - a_{ji})$.
It is well known that 7 invariants ${\rm tr}\, {\bf T}$, ${\rm tr}\, {\bf T}^2$, ${\rm tr}\, {\bf T}^3$, ${\rm tr}\, {\bf W}^2$, ${\rm tr}\, {\bf T} {\bf W}^2$, ${\rm tr}\, {\bf T}^2 {\bf W}^2$, ${\rm tr}\, {\bf T}^2 {\bf W}^2 {\bf T} {\bf W}$ form a minimal integrity basis of ${\bf A}(\mathcal{K})$ and also an irreducible function basis as well. We denote the invariants of ${\bf A}(\mathcal{K})$ as follows:
\begin{equation*}
\begin{array}{llll}
  I_1 :={\rm tr}\, {\bf T}, & I_2 := {\rm tr}\, {\bf T}^2, & J_2 := {\rm tr}\, {\bf W}^2, &I_3 := {\rm tr}\, {\bf T}^3,\\
  J_3 := {\rm tr}\, {\bf T} {\bf W}^2, & I_4 := {\rm tr}\, {\bf T}^2 {\bf W}^2, &I_6 := {\rm tr}\, {\bf T}^2 {\bf W}^2 {\bf T} {\bf W}.
\end{array}
\end{equation*}
The following theorem shows the way to obtain a minimal integrity basis of ${\cal K}$ from this particular minimal integrity basis of ${\bf A}({\cal K})$.

\begin{lemma}[Theorem 2.]\label{thm_integrity}
Let $\mathcal{K}$ be a Hall tensor with components $k_{ijk}$, and ${\bf A}(\mathcal{K})$ be its associated second order tensor with components $a_{ij}$. Denote $K_2 := I_1^2$, $J_4 := I_1I_3$, $K_4 := I_1J_3$, $J_6 := I_3^2$, $K_6 := J_3^2$, $L_6 :=I_3J_3$. Then the invariant set
\begin{align}\label{eq_integrity}
 \{I_2, J_2, K_2, I_4, J_4, K_4, I_6, J_6, K_6, L_6\}
\end{align}
is a minimal integrity basis of ${\mathcal K}$.
\end{lemma}
\begin{proof}
By Theorem \ref{thm_invariant}, any isotropic invariant of ${\mathcal K}$ is also an invariant of ${\bf A}(\mathcal{K})$, thus can be expressed by a polynomial $p(I_1, I_2, J_2, I_3, J_3, I_4, I_6)$.
Moreover, any isotropic invariant of an even order tensor consists of several even degree monomials.
Each even degree monomial containing $I_1, I_3, J_3$ should be a polynomial of $I_1^2, I_1I_3, I_1J_3, I_3^2, I_3J_3, J_3^2$.
Therefore, the isotropic invariant $p(I_1, I_2, J_2, I_3, J_3, I_4, I_6)$ can also be written into a polynomial of the invariants in \eqref{eq_integrity}.
That is, \eqref{eq_integrity} is an integrity basis of ${\mathcal K}$.

Next, we need to verify the polynomial irreducibility of this integrity basis. A natural observation is that these isotropic invariants are homogenous polynomials of the 9 independent components in the Hall tensor $\mathcal{K}$. A similar approach as the method proposed by Chen et al.\supercite{ChenQiZou17} is employed in this part.

\begin{enumerate}[(1)]
  \item There are exactly 3 degree-2 isotropic invariants $I_2, J_2, K_2$ in this integrity basis.
   Take $I_2$ for example.
   If it is not polynomial irreducible with the other 9 invariants in this basis, then it has to be a linear combination of the other 2 degree-2 invariants $J_2,K_2$.
   Therefore, if $I_2,J_2,K_2$ are polynomial irreducible, then the unique triple of $(c_1, c_2, c_3)$ such that
   \begin{align}\label{leq2}
   c_1I_2 + c_2J_2 + c_3K_2=0
   \end{align}
   is $c_1 = c_2 = c_3 =0$.
   Note that \eqref{leq2} holds for an arbitrary Hall tensor.
   Thus when we generate $n$ points $y_1, \cdots, y_{n}\in\mathbb{R}^{9}$, where $\mathbb{R}^{9}$ is the real number field with dimension 9, $c_1,c_2,c_3$ must be the solution to the linear system of equations
\begin{align}\label{lsy2}
\left(\begin{array}{ccc}
        I_2(y_1) & J_2(y_1) & K_2(y_1) \\
        I_2(y_2) & J_2(y_2) & K_2(y_2) \\
        \vdots  & \vdots & \vdots \\
        I_2(y_n) & J_2(y_n) & K_2(y_n)
      \end{array}\right)
      \left(\begin{array}{c}
      c_1 \\
      c_2 \\
      c_3
      \end{array}
      \right) =
      \left(\begin{array}{c}
      0 \\
      0 \\
      \vdots \\
      0
      \end{array}
      \right).
\end{align}
The coefficient matrix of System (\ref{lsy2}) is denoted by $M2$, and denote $r(M2)$ as the rank of the coefficient matrix $M2$. Then $r(M2)$ shows the number of polynomial irreducible invariants in these three isotropic invariants.
Take $n=3$ and
\begin{itemize}
  \item $y_1 =(-2, 3, 5, 0, -5, -4, -5, 2, -2)$,
  \item $y_2 =(-3, 0, 1, 1, 2, -4, 3, 0, 3)$,
  \item $y_3 =(-2,0, -1, 2, 1, -3, 5, 2, 3)$.
\end{itemize}
By numerical calculations, we can determine that $r(M2)=3$. Hence, the only solution for System (\ref{lsy2}) is $c_1 = c_2 = c_3 =0$, which implies that these three invariants of degree 2 are polynomial irreducible.

  \item For the invariants of degree $4$, we need to consider the following linear equation
\begin{align}\label{leq4}
&c_1(I_2)^2 + c_2(J_2)^2 + c_3(K_2)^2 + c_4I_2J_2 + c_5I_2K_2 + c_6J_2K_2 + c_7I_4 + c_8J_4 + c_9K_4=0,
\end{align}
where $c_1, \cdots, c_{9}$ are scalars.
If the unique $(c_1,c_2,\dots,c_9)$ such that \eqref{leq4} holds for any Hall tensor is $(0,0,\dots,0)$, then all the 3 degree-4 invariants $I_4,J_4,K_4$ are polynomial irreducible.
We generate $n$ points $y_1, \cdots, y_{n} \in \mathbb{R}^{9}$ and consider the following linear system:
\begin{align}\label{M4}
\left(\begin{array}{ccccc}
        I_2^2(y_1) & \cdots & I_4(y_1) & J_4(y_1) & K_4(y_1) \\
        I_2^2(y_2) & \cdots & I_4(y_2) & J_4(y_2) & K_4(y_2) \\
        \vdots  & \vdots & \vdots & \vdots & \vdots\\
        I_2^2(y_n) & \cdots & I_4(y_n) & J_4(y_n) & K_4(y_n)
      \end{array}\right)
      \left(\begin{array}{c}
      c_1 \\
      c_2 \\
      \vdots \\
      c_{9}
      \end{array}
      \right) =
      \left(\begin{array}{c}
      0 \\
      0 \\
      \vdots \\
      0
      \end{array}
      \right).
\end{align}
The coefficient matrix of System (\ref{M4}) is denoted by $M4$. Take $n=9$ and
 \begin{itemize}
 \item $y_1 =(4, 1, -3, 1, -4, -2, -1, 0, -5)$,
  \item $y_2 =(1, 5, 4, 0, -1, -5, -3, 5, -2)$,
  \item $y_3 =(-4, 4, -4, 1, -5, -2, 2, 3, 4)$,
  \item $y_4 =(-4, -5, 5, 5, -2, 3, 5, -1, 2)$,
  \item $y_5 =(0, 4, 3, 3, 1, -2, 3, 5, -4)$,
  \item $y_6 =(5, -3, 3, 3, -4, -2, 3, 5, -5)$,
  \item $y_7 =(-3, -2, 2, 4, -4, 1, 4, 2, 0)$,
  \item $y_8 = (-5, -3, 4, -1, 1, -2, -2, -3, 0)$,
  \item $y_9 =(0, -2, -2, 1, 5, 3, 4, 0, 0)$.
 \end{itemize}
 We can verify that the rank of $M4$ is $r(M4)= 9$, which implies that these three degree-4 invariants cannot be polynomial represented by other invariants of degree-4 and degree-2.

  \item Similarly, in the case of degree $6$, the verification linear equation is
\begin{align}\label{leq6}
&c_1(I_2)^3 + c_2(J_2)^3 + \cdots + c_{19}K_2K_4 + c_{20}I_6 + \cdots + c_{23}L_6=0.
\end{align}
Thus we generate $n$ points $y_1, \cdots, y_{n} \in \mathbb{R}^{9}$. Consider a linear system similar with system (\ref{M4}). Its coefficient matrix is denoted as $M6$, and its rank is denoted by $r(M6)$.
Take $n=23$ and
\begin{itemize}
  \item $y_1 =(3, -5, 1, 4, 2, 3, 3, 1, -3)$,
  \item $y_2 =(-5, -1, 2, -5, -2, 3, 3, 4, -1)$,
  \item $y_3 =(-4, 2, 1, -3, -2, -2, 1, 4, -1)$,
  \item $y_4 =(-2, 0, 3, 2, -2, -2, -5, 5, 2)$,
  \item $y_5 =(-2, -5, -5, -4, 3, -5, -3, 2, -3)$,
  \item $y_6 =(5, -4, 1, 3, -4, 1, -1, 4, 0)$,
  \item $y_7 =(-3, 3, 5, -3, -3, 1, 2, -2, -3)$,
  \item $y_8 = (2, 2, -5, 4, 4, -1, -5, 4, -5)$,
  \item $y_9 =(-2, -1, 2, 3, -2, -1, -2, -2, 5)$,
  \item $y_{10} =(-4, -3, -4, -2, -5, -5, 5, -2, -3)$,
  \item $y_{11} =(3, 2, -2, -5, 5, -3, 0, -2, -5)$,
  \item $y_{12} =(4, -4, -1, 4, -4, 0, 1, 3, -1)$,
  \item $y_{13} =(3, 0, -5, 0, 2, -5, -5, 4, 1)$,
  \item $y_{14} =(-4, 5, -5, 2, -1, -4, -5, -2, -5)$,
  \item $y_{15} =(2, -5, -5, 5, 0, 2, 2, 3, 4)$,
  \item $y_{16} =(1, 4, 4, -1, -5, -3, 4, -5, 1)$,
  \item $y_{17} =(-2, 5, -5, 1, -2, 1, 0, -5, 4)$,
  \item $y_{18} =(0, -4, -5, 0, -5, -2, -2, -2, 2)$,
  \item $y_{19} =(1, 2, 1, -1, 3, -4, -5, 4, 5)$,
  \item $y_{20} =(3, -3, 1, -3, -5, 3, 5, 1, 1)$,
  \item $y_{21} =(0, -1, 3, 0, -3, 5, 3, 0, 3)$,
  \item $y_{22} =(1, -5, -4, -1, 0, -1, -5, -5, 2)$,
  \item $y_{23} =(-4, -2, 3, 4, 5, -3, 4, 3, 3)$.
\end{itemize}
Then $r(M6)= 23$, which implies that these four invariants with degree $6$ are polynomial irreducible in the integrity basis.
\end{enumerate}
Therefore, we have shown that \eqref{eq_integrity} is a minimal integrity basis of ${\mathcal K}$.
\end{proof}

In the above discussion, we fix the inducing initial, i.e., a particular minimal integrity basis of the second order tensor. Nevertheless, the minimal integrity basis is generally not unique. We can also start from another minimal integrity basis of the second order tensor, denoted by
$$
\{\tilde{I}_1,\tilde{I}_2,\tilde{J}_2,\tilde{I}_3,\tilde{J}_3,\tilde{I}_4,\tilde{I}_6\}.
$$
Construct another integrity basis $\{\tilde{I}_2, \tilde{J}_2, \tilde{K}_2, \tilde{I}_4, \tilde{J}_4, \tilde{K}_4, \tilde{I}_6, \tilde{J}_6, \tilde{K}_6, \tilde{L}_6\}$ of the Hall tensor in the same way, where
$$
\tilde{K}_2:=\tilde{I}_1^2,\ \tilde{J}_4:=\tilde{I}_1\tilde{I}_3,\ \tilde{K}_4:=\tilde{I}_1\tilde{J}_3,\ \tilde{J}_6:=\tilde{I}_3^2,\ \tilde{K}_6:=\tilde{J}_3^2,\ \tilde{L}_6:=\tilde{I}_3\tilde{J}_3.
$$
Since this integrity basis has already got the same number of invariants as the minimal integrity basis \eqref{eq_integrity}, it must also be a minimal integrity basis. Therefore, we have the following corollary.

\begin{lemma}{\bf Corollary 1.}
  Let $\mathcal{K}$ be a Hall tensor with components $k_{ijk}$, and ${\bf A}(\mathcal{K})$ be its associated second order tensor with components $a_{ij}$. Let $\{\tilde{I}_1,\tilde{I}_2,\tilde{J}_2,\tilde{I}_3,\tilde{J}_3,\tilde{I}_4,\tilde{I}_6\}$ be any minimal integrity basis of the second order tensor ${\bf A}(\mathcal{K})$. Denote $\Psi := \{\tilde{I}_2, \tilde{J}_2, \tilde{K}_2, \tilde{I}_4, \tilde{J}_4, \tilde{K}_4, \tilde{I}_6, \tilde{J}_6, \tilde{K}_6, \tilde{L}_6\}$ with $\tilde{K}_2:=\tilde{I}_1^2$, $\tilde{J}_4:=\tilde{I}_1\tilde{I}_3$, $\tilde{K}_4:=\tilde{I}_1\tilde{J}_3$, $\tilde{J}_6:=\tilde{I}_3^2$, $\tilde{K}_6:=\tilde{J}_3^2$, $\tilde{L}_6:=\tilde{I}_3\tilde{J}_3$. Then $\Psi$ is a minimal integrity basis of the Hall tensor $\mathcal{K}$.
\end{lemma}

\section{Irreducible function basis}

Since a minimal integrity basis for a tensor is also a function basis, the number of invariants in an irreducible function basis consisting of polynomial invariants is no more than that of a minimal integrity basis. Moreover, the number of invariants in a minimal integrity basis of a tensor can be very big. For example, the number of minimal integrity basis of an elasticity tensor is 297\supercite{OliveKolevAuffray17}. However, from an experimental point of view, it will be easier to detect all the values of the invariants in an irreducible function basis of a tensor. Hence, it is meaningful to study the irreducible function basis of a tensor. For a symmetric third order tensor, one of its irreducible function base contains 11 invariants, while its minimal integrity basis contains 13 invariants\supercite{CLQZZ18}.

In this section, we shall show that the minimal integrity basis given in Section 3 is also an irreducible function basis of the Hall tensor ${\mathcal K}$.
According to the method proposed by Pennisi and Trovato\supercite{Pennisi87} in 1987,  to show a given function basis of a tensor is functionally irreducible, for each invariant in this basis, we need to find two different sets of independent variables in the tensor, denoted by $V$ and $V^{'}$, such that this invariant takes different values in $V$ and $V^{'}$ while all the remainders are the same in $V$ and $V^{'}$.  The following theorem is proved in this spirit.

\begin{lemma}[Theorem 3]\label{thm_function}
The set  $\{I_2, J_2, K_2, I_4, J_4, K_4, I_6, J_6, K_6, L_6\}$ is an irreducible function basis of the Hall tensor.
\end{lemma}
\begin{proof}
It can be verified by definitions  that an integrity basis of a tensor is a function basis of the tensor. We have proved in Section 3 that these ten invariants form a minimal integrity basis of the Hall tensor. Thus this basis is also a function basis.

Denote $V=\{k_{121}, k_{122}, k_{123}, k_{131}, k_{132}, k_{133}, k_{231}, k_{232}, k_{233}\}$ and $V^{'}=\{k_{121}^{'}, k_{122}^{'}, k_{123}^{'}, k_{131}^{'},\\ k_{132}^{'}, k_{133}^{'}, k_{231}^{'}, k_{232}^{'}, k_{233}^{'}\}$ as two different sets of independent variables of the Hall tensor $\mathcal{K}$. Then we shall find ten pairs of $\{V, V^{'}\}$ to show that all the ten isotropic invariants in  \eqref{eq_integrity} is functionally irreducible.

\begin{enumerate}[{\rm (1)}]

\item For $I_2$, in $V$, let $k_{121} = k_{122} = k_{123} = k_{132} = k_{133} = k_{231} = k_{233} = 0,$ $ k_{131}= -1, k_{232}= 1$.

Then in $V^{'}$, let $k_{121}^{'} = k_{122}^{'} = k_{123}^{'} = k_{132}^{'} = k_{133}^{'} = k_{231}^{'} = k_{233}^{'} = 0,$ $ k_{131}^{'}= -2, k_{232}^{'}= 2$.

We have that $I_2 = 2$ and $I_2^{'} = 8$, while other invariants: $\{ J_2, K_2, I_4, J_4, K_4, I_6, J_6, K_6, L_6\}$, and $\{J_2^{'}, K_2^{'}, I_4^{'}, J_4^{'}, K_4^{'}, I_6^{'}, J_6^{'}, K_6^{'}, L_6^{'}\}$ are all equal to $0$. This means that $I_2$ is functionally irreducible in the function basis  \eqref{eq_integrity}.

\item For $J_2$, in $V$, let $k_{121} = k_{122} = k_{123} = k_{132} = k_{133} = k_{231} = k_{233} = 0,$ $ k_{131}= 1, k_{232}= 1$.
Then in $V^{'}$, let all the variables be $0$.

We have that $J_2 = 2$ and $J_2^{'} = 0$, while other invariants: $\{ I_2, K_2, I_4, J_4, K_4, I_6, J_6, K_6, L_6\}$, and $\{I_2^{'},  K_2^{'}, I_4^{'}, J_4^{'}, K_4^{'}, I_6^{'}, J_6^{'}, K_6^{'}, L_6^{'}\}$ are all equal to $0$. This means that $J_2$ is functionally irreducible in the function basis  \eqref{eq_integrity}.

\item For $K_2$, in $V$, let $k_{121} = k_{122} = k_{131}  = k_{133} = k_{232} = k_{233} = 0,$ and $ k_{123}= - \sqrt{\frac{2+\sqrt[3]{4}}{2}}, k_{132}= 0,  k_{231}= \sqrt{\frac{2+\sqrt[3]{4}}{2}}$.

In $V^{'}$, let $k_{121}^{'} = k_{122}^{'} = k_{131}^{'}  = k_{133}^{'} = k_{232}^{'} = k_{233}^{'} = 0,$ and $ k_{123}^{'}= 1, k_{132}^{'}= \sqrt[3]{2},  k_{231}^{'}= 1$.

We have $K_2 = 0$. It is not equal to $K_2^{'} = (2 - \sqrt[3]{2})^2$, while other invariants: $I_2 = I_2^{'}= 2 + \sqrt[3]{4}$,  and $J_2 = J_2^{'}= I_4 = I_4^{'} = J_4 = J_4^{'} = K_4 = K_4^{'} = I_6 = I_6^{'} = J_6 = J_6^{'} = K_6 = K_6^{'} = L_6 = L_6^{'} = 0$. This means that $K_2$ is functionally irreducible.

\item For $I_4$, in $V$, let $k_{121}=-2,$ $k_{122}=0,$ $k_{123}=1,$  $k_{131}=1,$ $k_{132}=1,$ $k_{133}=0,$ $k_{231}=0,$ $k_{232}=1,$ and $k_{233}=2$.

In $V^{'}$, let $k_{121}^{'}=-\sqrt{3},$ $k_{122}^{'}=-\sqrt{2},$ $k_{123}^{'}=1,$  $k_{131}^{'}=0,$ $k_{132}^{'}=1,$ $k_{133}^{'}=-\sqrt{2},$ $k_{231}^{'}=0,$ $k_{232}^{'}=0,$ and $k_{233}^{'}=\sqrt{3}$.

We have $I_4 = 5$. It is not equal to $I_4^{'} = 7,$ while $I_2 = I_2^{'}= 2 $, $J_2 = J_2^{'}= 10$, $K_6 = K_6^{'}= 9$, and others are all equal to 0. This means that $I_4$ is functionally irreducible.

\item For $J_4$, assume that $s=4+\sqrt{14}$, and $t=4-\sqrt{14}$. In $V$, let $k_{121} = k_{122} = k_{131}  = k_{133} = k_{232} = k_{233} = 0,$ and $$ k_{123}= 1, k_{132}= 1,  k_{231}=\frac{\sqrt[3]{2t}}{2} + \sqrt[3]{\frac{s}{4}}.$$

In $V^{'}$, let $k_{121}^{'} = k_{122}^{'} = k_{131}^{'}  = k_{133}^{'} = k_{232}^{'} = k_{233}^{'} = 0,$ and
\begin{align*}
  k_{123}^{'} &= 2- \frac{\sqrt[3]{2t}}{2}- \frac{\sqrt[3]{2s}}{2} -\frac{\sqrt[6]{2}}{2}\sqrt{2\sqrt[3]{4}+8\sqrt[3]{t}+\sqrt[3]{2t^2}+8\sqrt[3]{s}+\sqrt[3]{2s^2}},\\
  k_{132}^{'} &= -1 + \frac{\sqrt[3]{2t}}{4}+ \frac{\sqrt[3]{2s}}{4} -\frac{\sqrt[6]{2}}{4}\sqrt{2\sqrt[3]{4}+8\sqrt[3]{t}+\sqrt[3]{2t^2}+8\sqrt[3]{s}+\sqrt[3]{2s^2}},
\end{align*}
and $k_{231}^{'}= 0.$

We have $J_4 = - J_4^{'} =\frac{3}{8}\left(-4 + \sqrt[3]{2t} +\sqrt[3]{2s}\right)\left(\sqrt[3]{2t} +\sqrt[3]{2s}\right).$
Meanwhile,
\begin{equation*}
\begin{array}{lll}
I_2 &= I_2^{'} &= 2 + \frac{\left(\sqrt[3]{2t}+ \sqrt[3]{2s}\right)^2}{4},\\
K_2 &= K_2^{'} &= \frac{1}{4}\left(-4 + \sqrt[3]{2t} +\sqrt[3]{2s}\right)^2,\\
J_6 &= J_6^{'} &= \frac{9}{16}\left(\sqrt[3]{2t} +\sqrt[3]{2s}\right)^2,
\end{array}
\end{equation*}
 and others are all equal to 0. This shows that $J_4$ is functionally irreducible.

\item For $K_4$, in $V$, let
\begin{equation*}
\begin{array}{lll}
k_{121}=-\frac{1}{2}\sqrt{\frac{-12+6\sqrt[3]{9}}{16-3\sqrt[3]{3}-3\sqrt[3]{9}}}, &k_{122}=\frac{1}{2}, &k_{123}=-1,\\
k_{131}=0, &k_{132}=-\frac{\sqrt[3]{9}}{2}, &k_{133}=\frac{1}{2},\\
k_{231}=-\frac{1}{2}, &k_{232}=0, &k_{233}=\frac{1}{2}\sqrt{\frac{-12+6\sqrt[3]{9}}{16-3\sqrt[3]{3}-3\sqrt[3]{9}}}.
\end{array}
\end{equation*}

In $V^{'}$, let
\begin{equation*}
\begin{array}{lll}
k_{121}^{'}=0, &k_{122}^{'}=-\frac{\sqrt[6]{3}}{2}\sqrt{\frac{9+5\sqrt[3]{3}-6\sqrt[3]{9}}{16- 3\sqrt[3]{3}-3\sqrt[3]{9}}}, &k_{123}^{'}=1,\\
k_{131}^{'}=-\frac{1}{2}\sqrt{\frac{22 - 12\sqrt[3]{3}- 2\sqrt[3]{9}}{16- 3\sqrt[3]{3}-3\sqrt[3]{9}}}, &k_{132}^{'}=\frac{\sqrt[3]{9}}{2}, &k_{133}^{'}=-\frac{\sqrt[6]{3}}{2}\sqrt{\frac{9+5\sqrt[3]{3}-6\sqrt[3]{9}}{16- 3\sqrt[3]{3}-3\sqrt[3]{9}}},\\
k_{231}^{'}=\frac{1}{2}, &k_{232}^{'}=-\frac{1}{2}\sqrt{\frac{22 - 12\sqrt[3]{3}- 2\sqrt[3]{9}}{16- 3\sqrt[3]{3}-3\sqrt[3]{9}}}, &k_{233}^{'}=0.
\end{array}
\end{equation*}

We have
\begin{equation*}
K_4 = - K_4^{'} = \frac{6 + 21\sqrt[3]{3}- 17\sqrt[3]{9}}{-256 + 48\sqrt[3]{3}+48\sqrt[3]{9}}.
\end{equation*}

Meanwhile,\vspace{2mm}
\begin{equation*}
\begin{array}{lll}
I_2 &= I_2^{'} &= \frac{5+3\sqrt[3]{3}}{4},\\ \vspace{2mm}
J_2 &= J_2^{'} &= \frac{4 - 3\sqrt[3]{3}+ 3\sqrt[3]{9}}{-32 + 6\sqrt[3]{3}+ 6\sqrt[3]{9}},\\ \vspace{2mm}
K_2 &= K_2^{'} &= \frac{(-3 + \sqrt[3]{9})^2}{4},\\ \vspace{2mm}
I_4 &= I_4^{'} &= \frac{-23 + 36\sqrt[3]{3} + 9\sqrt[3]{9}}{-256 + 48\sqrt[3]{3}+48\sqrt[3]{9}},\\ \vspace{2mm}
K_6 &= K_6^{'} &= \frac{-47 + 78\sqrt[3]{3} - 31\sqrt[3]{9}}{64(-16 + 3\sqrt[3]{3}+ 3\sqrt[3]{9})^2},
\end{array}
\end{equation*}
 and others are all equal to 0. This shows that $K_4$ is functionally irreducible.

\item For $I_6$, in $V$, let $k_{121}= -1,$ $k_{122}=-1,$ $k_{123}= 1,$  $k_{131}= 1,$ $k_{132}= 1,$ $k_{133}= -1,$ $k_{231}=0,$ $k_{232}=1,$ and $k_{233}=1$.

In $V^{'}$, let $k_{121}^{'}= -1,$ $k_{122}^{'}= -1,$ $k_{123}^{'}= 1,$  $k_{131}^{'}= -1,$ $k_{132}^{'}= 1,$ $k_{133}^{'}= -1,$ $k_{231}^{'}=0,$ $k_{232}^{'}= -1,$ and $k_{233}^{'}= 1$.

We have $I_6 = - I_6^{'} =2$. Meanwhile, $I_2 = I_2^{'}= 2 $, $J_2 = J_2^{'}= 6$, $I_4 = I_4^{'}= 4$, and others are all equal to 0. This means that $I_6$ is functionally irreducible.

\item For $J_6$, in $V$, let $k_{121} = k_{122} = k_{131}  = k_{133} = k_{232} = k_{233} = 0,$ and $ k_{123}= -\sqrt{3}\sqrt[3]{2}, k_{132}= 0,  k_{231}= \sqrt{3}\sqrt[3]{2}$.

In $V^{'}$, let $k_{121}^{'} = k_{122}^{'} = k_{131}^{'}  = k_{133}^{'} = k_{232}^{'} = k_{233}^{'} = 0,$ and $ k_{123}^{'}= \sqrt[3]{2}, k_{132}^{'}= 2\sqrt[3]{2},  k_{231}^{'}= \sqrt[3]{2}$.

We have $J_6 = 0 \neq J_6^{'} = 144$. Meanwhile, $I_2 = I_2^{'}= 6\sqrt[3]{4} $,  and others are all equal to 0. This means that $J_6$ is functionally irreducible.

\item For $K_6$, in $V$, let $k_{121}= \frac{1}{2},$ $k_{122}=1,$ $k_{123}=0,$  $k_{131}=\frac{3}{2},$ $k_{132}=0,$ $k_{133}=1,$ $k_{231}=0,$ $k_{232}=\frac{3}{2},$ and $k_{233}=\frac{1}{2}$.

In $V^{'}$, let $k_{121}^{'}=-\frac{1}{2},$ $k_{122}^{'}=\frac{1}{2},$ $k_{123}^{'}=0,$  $k_{131}^{'}=\sqrt{3},$ $k_{132}^{'}=0,$ $k_{133}^{'}=\frac{1}{2},$ $k_{231}^{'}=0,$ $k_{232}^{'}=\sqrt{3},$ and $k_{233}^{'}=-\frac{1}{2}$.

We have $K_6 = \frac{9}{4} \neq K_6^{'} = \frac{3}{4}.$  Meanwhile, $I_2 = I_2^{'}= \frac{1}{2}$, $J_2 = J_2^{'}= - \frac{13}{2}$, $I_4 = I_4^{'}= -\frac{13}{16}$, and others are all equal to 0. This means that $K_6$ is functionally irreducible.

\item For $L_6$, in $V$, let $k_{121}= -1,$ $k_{122}=\frac{1}{2},$ $k_{123}=-1,$  $k_{131}=0,$ $k_{132}=2,$ $k_{133}=\frac{1}{2},$ $k_{231}=3,$ $k_{232}=0,$ and $k_{233}=1$.

In $V^{'}$, let $k_{121}^{'}=0,$ $k_{122}^{'}=-\frac{1}{2}\sqrt{\frac{5}{2}},$ $k_{123}^{'}=1,$  $k_{131}^{'}=-\frac{1}{2}\sqrt{\frac{5}{2}},$ $k_{132}^{'}=-2,$ $k_{133}^{'}=-\frac{1}{2}\sqrt{\frac{5}{2}},$ $k_{231}^{'}=-3,$ $k_{232}^{'}=-\frac{1}{2}\sqrt{\frac{5}{2}},$ and $k_{233}^{'}= 0$.

We have $L_6 = - L_6^{'} = - \frac{45}{2}.$  Meanwhile, $I_2 = I_2^{'}= 14$, $J_2 = J_2^{'}= - \frac{5}{2}$, $I_4 = I_4^{'}= -\frac{45}{4}$, $J_6 = J_6^{'}= 324$, $K_6 = K_6^{'}= \frac{25}{16}$, and others are all equal to 0. This means that $L_6$ is functionally irreducible.
\end{enumerate}

Therefore, this particular minimal integrity basis  $\{I_2, J_2, K_2, I_4, J_4, K_4, I_6, J_6, K_6, L_6\}$ is also an irreducible function basis of the Hall tensor $\mathcal{K}$.
\end{proof}

In the above proof, the examples $V$ and $V^{'}$ in the cases (1), (2), (4) and (7) are based on related sets in Pennisi and Trovato\supercite{Pennisi87}, while the examples $V$ and $V^{'}$ in the case (5) are suggested by Dr. Yannan Chen.

\section{Conclusions and A Further Question}

In this paper, we investigate isotropic  invariants of the Hall tensor.
For this purpose, we connect the invariants of the Hall tensor ${\mathcal K}$ with the ones of its associated second order tensor ${\bf A}({\cal K})$. ${\bf A}({\cal K})$ can be split into a second order symmetric tensor ${\bf T}$ and a second order skew-symmetric tensor ${\bf W}$.
Then $\{ I_1 := {\rm tr}\, {\bf T}, I_2 := {\rm tr}\, {\bf T}^2, J_2 := {\rm tr}\, {\bf W}^2, I_3 := {\rm tr}\, {\bf T}^3, J_3 := {\rm tr}\, {\bf T} {\bf W}^2, I_4 := {\rm tr}\, {\bf T}^2 {\bf W}^2, I_6 := {\rm tr}\, {\bf T}^2 {\bf W}^2 {\bf T} {\bf W}\}$ is the minimal integrity basis of ${\bf A}({\mathcal K})$ as in the previous sections. It is also an irreducible function basis of  ${\bf A}({\mathcal K})$.
We prove in this paper the following statements:
\begin{enumerate}[(i)]
  \item $\{ I_1^2, I_2, J_2, I_4, I_1I_3, I_1J_3, I_6, I_3^2, J_3^2, I_3J_3\}$ is an isotropic minimal  integrity basis of the Hall tensor ${\cal K}$.
  \item $\{I_1^2, I_2, J_2, I_4, I_1I_3, I_1J_3, I_6, I_3^2, J_3^2, I_3J_3 \}$ is also an isotropic irreducible function basis of the Hall tensor ${\cal K}$ as well.
\end{enumerate}

Apart from this particular selection, we can also begin with any minimal integrity basis of the second order tensor and use the same approach to construct an invariant basis of the Hall tensor. We prove in the paper that such basis of the Hall tensor is a minimal integrity basis.

A further question is whether there exists an irreducible function basis consisting of less than ten polynomial invariants.

\bigskip

\noindent  \textbf{Acknowledgements}\quad\footnotesize   We are thankful to Dr. Yannan Chen for his helpful discussions.   


\begin{thebibliography}{00}   
\footnotesize
\bibitem{Boehler1977}
Boehler J. P. On irreducible representations for isotropic scalar functions. \emph{Journal of Applied Mathematics and Mechanics}, \textbf{57}(6), 323--327 (1977)

\bibitem{BKO94}
Boehler J. P., Kirillov A. A., and Onat E. T. On the polynomial invariants of the elasticity tensor. \emph{Journal of elasticity}, \textbf{34}(2), 97--110 (1994)

\bibitem{Pennisi87}
Pennisi S. and Trovato M. On the irreducibility of Professor G.F. Smith's representations for isotropic functions. \emph{International journal of engineering science}, \textbf{25}(8), 1059--1065 (1987)

\bibitem{Smith71}
Smith G. F. On isotropic functions of symmetric tensors, skew-symmetric tensors and vectors. \emph{International Journal of Engineering Science}, \textbf{9}(10), 899-916 (1971)

\bibitem{SmithBao1997}
Smith G. F. and Bao G. Isotropic invariants of traceless symmetric tensors of orders three and four. \emph{International Journal of Engineering Science}, \textbf{35}(15), 1457--1462 (1997)

\bibitem{Spencer1970}
Spencer A. J. M. A note on the decomposition of tensors into traceless symmetric tensors. \emph{International Journal of Engineering Science}, \textbf{8}(6), 475--481 (1970)

\bibitem{Wang70}
Wang C. C. A new representation theorem for isotropic functions: An answer to Professor G.F. Smith's criticism of my papers on representations for isotropic functions. \emph{Archive for rational mechanics and analysis}, \textbf{36}(3), 166--197 (1970)


\bibitem{Zheng93}
Zheng Q. S. On the representations for isotropic vector-valued, symmetric tensor-valued and skew-symmetric tensor-valued functions. \emph{International journal of engineering science}, \textbf{31}(7), 1013--1024 (1993)

\bibitem{Zheng94}
Zheng Q. S. Theory of representations for tensor functions - a unified invariant approach to constitutive equations. \emph{Applied Mechanics Reviews }, \textbf{47}(11), 545--587 (1994)


\bibitem{Zheng96}
Zheng Q. S. Two-dimensional tensor function representations involving third-order tensors. \emph{ Archives of Mechanics}, \textbf{48}(4), 659--673 (1996)

\bibitem{ZhengBetten95}
Zheng Q. S. and Betten J. On the tensor function representations of {$2$}nd-order and {$4$}th-order tensors. Part {I}. \emph{ Journal of Applied Mathematics and Mechanics}, \textbf{75}(4), 269--281 (1995)

\bibitem{ChenQiZou17}
Chen Y., Qi L., and Zou W. Irreducible function bases of isotropic invariants of third and fourth order symmetric tensors. \emph{ arXiv:1712.02087v3, Preprint,} (2018)

\bibitem{CLQZZ18}
Chen Z., Liu J., Qi L., Zou W., and Zheng Q. S. An irreducible function basis of isotropic invariants of a third order three-dimensional symmetric tensor. \emph{arXiv:1803.01269, Preprint,} (2018)

\bibitem{AuffrayRopars16}
Auffray N. and Ropars P. Invariant-based reconstruction of bidimensional elasticity tensors. \emph{ International Journal of Solids and Structures}, \textbf{87}, 183--193 (2016)

\bibitem{Hilbert93}
Hilbert D. \emph{Theory of algebraic invariants}. Cambridge university press (1993)


\bibitem{OliveAuffray14}
Olive M. and Auffray N. Isotropic invariants of a completely symmetric third-order tensor. \emph{ Journal of Mathematical Physics}, \textbf{55}(9), 092901 (2014)

\bibitem{Sturmfels}
Sturmfels B. \emph{Algorithms in invariant theory}. Springer Science \& Business Media (2008)

\bibitem{OliveKolevAuffray17}
Olive M., Kolev B., and Auffray N. A minimal integrity basis for the elasticity tensor. \emph{Archive for Rational Mechanics and Analysis}, \textbf{226}(1), 1--31 (2017)


\bibitem{Haussuhl08}
Hauss{\"u}hl S. \emph{Physical Properties of Crystals: An Introduction}, John Wiley \& Sons, 205--207 (2008)

\bibitem{Hall1879}
Hall, E. On a new action of the magnet on electric currents.  \emph{American Journal of Mathematics}, \textbf{2}(3), 287--92 (1879)




\end{thebibliography}
\end{document}